\documentclass[copyright,creativecommons]{eptcs}
\providecommand{\event}{NCL'22} 
\usepackage{underscore}           
\usepackage{amsmath}
\usepackage{amsthm}
\usepackage{amsfonts}
\usepackage{amssymb}
\usepackage{enumerate}
\usepackage{hyperref}
\usepackage{quiver}
\usepackage[all,cmtip]{xy}
\usepackage{proof}

\newtheorem{theorem}{Theorem}[section]

\theoremstyle{definition}
\newtheorem{defn}{Definition}[section]
\newtheorem{example}{Example}[section]

\makeatletter
\newsavebox{\@brx}
\newcommand{\llangle}[1][]{\savebox{\@brx}{\(\m@th{#1\langle}\)}%
  \mathopen{\copy\@brx\kern-0.5\wd\@brx\usebox{\@brx}}}
\newcommand{\rrangle}[1][]{\savebox{\@brx}{\(\m@th{#1\rangle}\)}%
  \mathclose{\copy\@brx\kern-0.5\wd\@brx\usebox{\@brx}}}
\makeatother
\newcommand{\ldbc}{[\![}
\newcommand{\rdbc}{]\!]}

\newcommand{\tl}{\otimes \mathsf{L}}
\newcommand{\tr}{\otimes \mathsf{R}}
\newcommand{\lright}{{\multimap}\mathsf{R}}
\newcommand{\lleft}{{\multimap}\mathsf{L}}
\newcommand{\pass}{\mathsf{pass}}
\newcommand{\unitl}{\mathsf{IL}}
\newcommand{\unitr}{\mathsf{IR}}
\newcommand{\ax}{\mathsf{ax}}
\newcommand{\id}{\mathsf{id}}
\newcommand{\ot}{\otimes}
\newcommand{\lolli}{\multimap}
\newcommand{\illol}{\rotatebox[origin=c]{180}{$\multimap$}}
\newcommand{\I}{\mathsf{I}}

\newcommand{\RI}{\mathsf{RI}}
\newcommand{\LI}{\mathsf{LI}}
\newcommand{\Pass}{\mathsf{P}}
\newcommand{\F}{\mathsf{F}}
\newcommand{\xvdash}{\vdash^{x}}

\newcommand{\otd}{\ot^D}
\newcommand{\lollid}{\;\textsuperscript{$D$}\!\!\lolli}

\newcommand{\proofbox}[1]{\begin{tabular}{l} #1 \end{tabular}}

\newcommand{\MILL}{$\mathtt{MILL}$}
\newcommand{\NMILL}{$\mathtt{NMILL}$}
\newcommand{\SkNMILL}{$\mathtt{SkNMILL}$}
\newcommand{\FSkMCC}{\mathsf{FSkMCl}}

\title{Proof Theory of Skew Non-Commutative \MILL}
\author{Tarmo Uustalu
\institute{Reykjavik University, Iceland}
\institute{Tallinn University of Technology, Estonia}
\email{tarmo@ru.is}
\and
Niccol{\`o} Veltri \qquad\qquad Cheng-Syuan Wan
\institute{Tallinn University of Technology, Estonia}
\email{\quad niccolo@cs.ioc.ee \quad\qquad cswan@cs.ioc.ee}
}
\def\titlerunning{Proof Theory of Skew Non-Commutative \MILL}
\def\authorrunning{T. Uustalu, N. Veltri \& C.-S. Wan}
\begin{document}
\maketitle
\begin{abstract}
  Monoidal closed categories naturally model \NMILL, non-commutative multiplicative intuitionistic linear logic: the monoidal unit and tensor interpret the multiplicative verum and conjunction; the internal hom interprets linear implication. In recent years, the weaker notion of (left) skew monoidal closed category has been proposed by Ross Street, where the three structural laws of left and right unitality and associativity are not required to be invertible, they are merely natural transformations with a specific orientation. A question arises: is it possible to find a logic which is naturally modelled by skew monoidal closed categories? We answer positively by introducing a cut-free sequent calculus for a skew version of \NMILL\ that is a presentation of the free skew monoidal closed category. We study the proof-theoretic semantics of the sequent calculus by identifying a calculus of derivations in normal form, obtained from an adaptation of Andreoli's focusing technique to the skew setting. The resulting focused sequent calculus peculiarly employs a system of tags for keeping track of new formulae appearing in the antecedent and appropriately reducing non-deterministic choices in proof search. Focusing solves the coherence problem for skew monoidal closed categories by exhibiting an effective procedure for deciding equality of maps in the free such category.
\end{abstract}

\section{Introduction}

It is a widely known fact from the late 80s/early 90s that symmetric monoidal closed categories model \MILL, multiplicative intuitionistic linear logic, whose logical connectives comprise multiplicative verum $\I$ and conjunction $\ot$ and linear implication $\lolli$ \cite{mellies:categorical:09}. Probably lesser-known, though a quite straightforward variation of this observation (which actually predated the inception of linear logic altogether) is the fact that (not necessarily symmetric) monoidal biclosed categories provide a semantics for the non-commutative variant of \MILL\ (for which we use the abbreviation \NMILL) \cite{abrusci:noncommutative:1990}, where the structural rule of exchange is absent and there are two ordered implications $\lolli$ and $\illol$. The sequent calculus of \NMILL\ without verum is known as the \emph{Lambek calculus} \cite{lambek:deductive:68}. From the beginning, the Lambek calculus has been employed in formal investigations of natural languages \cite{lambek:mathematics:58}. Lambek refers to the implications $\lolli$ and $\illol$ as residuals, and monoidal biclosed categories without unit are therefore also called residual categories.
By dropping one of the ordered implications of \NMILL, one obtains a fragment of the logic interpretable in every monoidal closed category.

In recent years, Ross Street introduced the new notion of (left) skew monoidal closed categories \cite{street:skew-closed:2013}. These are a weakening of monoidal closed categories: their structure includes a unit $\I$, a tensor $\ot$, an internal hom $\lolli$ and an adjunction relating the latter two operations, as in usual non-skew monoidal closed categories. The difference lies in the three structural laws of left and right unitality, $\lambda_A : \I \ot A \to A$ and $\rho_A : A \to A \ot \I$, and associativity, $\alpha_{A,B,C} : (A \ot B) \ot C \to A \ot (B \ot C)$, which are usually required to be natural isomorphisms, but in the skew variant are merely natural transformations with the specified orientation. Street originally proposed this weaker notion to reach better understanding of and to fix a famous imbalance, first noticed by Eilenberg and Kelly \cite{eilenberg:closed:1966}, present in the adjunction relating monoidal and closed structures \cite{street:skew-closed:2013,uustalu:eilenberg-kelly:2020}. In the last decade, skew monoidal closed categories, together with their non-closed/non-monoidal variants, have been thoroughly studied, with applications ranging from algebra and homotopy theory to programming language semantics \cite{szlachanyi:skew-monoidal:2012,lack:skew:2012,lack:triangulations:2014,altenkirch:monads:2014,buckley:catalan:2015,bourke:skew:2017,bourke:skew:2018,tomita:realizability:21}.

A question arises naturally: is it possible to characterize skew monoidal closed categories as categorical models of (a deductive system for) a logic?
This paper provides a positive answer to this question. We introduce a cut-free sequent calculus for a skew variant of \NMILL, which we name \SkNMILL. Sequents are peculiarly defined as triples $S \mid \Gamma \vdash A$, where the succedent is a single formula $A$ (as in intuitionistic linear logic), but the antecedent is divided into two parts: an optional formula $S$, called the \emph{stoup} \cite{girard:constructive:91}, and an ordered list of formulae $\Gamma$, the \emph{context}. Inference rules are similar to the ones of \NMILL\ but with specific structural restrictions for accurately capturing the structural laws of skew monoidal closed categories and nothing more. In particular, and in analogy with \NMILL, the structural rules of weakening, contraction and exchange are all absent. Sets of derivations are quotiented by a congruence relation $\circeq$, carefully crafted to serve as the sequent-calculus counterpart of the equational theory of skew monoidal closed categories. The design of the sequent calculus draws inspiration from, and further advance, the line of work of Uustalu, Veltri and Zeilberger on proof systems for various categories with skew structure: skew semigroup (a.k.a.\ the Tamari order) \cite{zeilberger:semiassociative:19}, skew monoidal (non-closed) \cite{uustalu:sequent:2021,uustalu:proof:nodate} and its symmetric variant \cite{veltri:coherence:2021}, skew prounital (non-monoidal) closed \cite{uustalu:deductive:nodate}.

The metatheory of \SkNMILL\ is developed in two different but related directions:
\begin{itemize}
  \item We study the categorical semantics of \SkNMILL, by showing that the cut-free sequent calculus admits an interpretation of its formulae, derivations and equational theory $\circeq$ in any skew monoidal closed category as soon as one fixes the intended interpretation in it of the atoms. Moreover, the sequent calculus is the initial model in this semantics; it is a particular presentation of the free skew monoidal closed category. This can be shown directly or by first introducing a Hilbert-style calculus which directly presents the free skew monoidal closed category and proving that derivations in the two calculi are in a bijective correspondence.

\item We investigate the proof-theoretic semantics of \SkNMILL, by defining a normalization strategy for sequent calculus derivations wrt.\ the congruence $\circeq$, when the latter is considered as a locally confluent and strongly normalizing reduction relation. The shape of normal forms is made explicit in a new \emph{focused} sequent calculus, whose derivations act as the target of the normalization procedure. The sequent calculus is ``focused'' in the sense of Andreoli \cite{andreoli:logic:1992}, as it describes a sound and complete root-first proof search strategy for the original sequent calculus. The focused system in this paper builds on and elaborates the previously
developed normal forms for skew monoidal categories \cite{uustalu:sequent:2021} and skew prounital closed categories \cite{uustalu:deductive:nodate}. The presence of both positive ($\I$,$\ot$) and negative ($\lolli$) connectives requires some extra care in the design of the proof search strategy, reflected in the focused sequent calculus, in particular when aiming at removing all possible undesired non-deterministic choices leading to $\circeq$-equivalent derivations that can arise during proof search. This is technically realized in the focused sequent calculus by peculiar employment of a system of \emph{tags} for keeping track of new formulae appearing in the antecedent. The focused sequent calculus can also be seen as an optimized presentation of the initial model for \SkNMILL, and as such can be used for solving the coherence problem in an effective way: deciding equality of two maps in this model (or equality of two canonical maps in every model) is equivalent to deciding \emph{syntactic} equality of the corresponding focused derivations.
\end{itemize}

The equivalence between sequent calculus derivations, quotiented by the equivalence relation $\circeq$, and focused derivations, that we present in Section \ref{sec:focus}, has been formalized in the Agda proof assistant. The associated code is available at \url{https://github.com/niccoloveltri/code-skewmonclosed}.

\section{A Sequent Calculus for Skew Non-Commutative \MILL}\label{sec2}

We begin by introducing a sequent calculus for a skew variant of non-commutative multiplicative intuitionistic linear logic (\NMILL), that we call \SkNMILL.

Formulae are inductively generated by the grammar $A,B ::= X \ | \ \I \ | \ A \ot B \ | \ A \lolli B$, where $X$ comes from a fixed set $\mathsf{At}$ of atoms, $\I$ is a multiplicative verum, $\ot$ is a multiplicative conjunction and $\lolli$ is a linear implication.

A sequent is a triple of the form $S \mid \Gamma \vdash A$, where the succedent $A$ is a single formula (as in \NMILL) and the antecedent is divided in two parts: an optional formula $S$, called \emph{stoup} \cite{girard:constructive:91}, and an ordered list of formulae $\Gamma$, called \emph{context}. The peculiar design of sequents, involving the presence of the stoup in the antecedent, comes from previous work on deductive systems with skew structure by Uustalu, Veltri and Zeilberger \cite{uustalu:sequent:2021,uustalu:proof:nodate,uustalu:deductive:nodate,veltri:coherence:2021}.
The metavariable $S$ always denotes a stoup, i.e., $S$ can be a single formula or empty, in which case we write $S = {-}$, and $X,Y,Z$ are always names of atomic formulae.

Derivations of a sequent $S \mid \Gamma \vdash A$ are inductively generated by the following rules:

\begin{equation}\label{eq:seqcalc}
  \def\arraystretch{2.5}
  \begin{array}{c}
    \infer[\ax]{A \mid \quad \vdash A}{}
    \qquad
    \infer[\pass]{{-} \mid A , \Gamma \vdash C}{A \mid \Gamma \vdash C}
    \qquad
    \infer[\lleft]{A \lolli B \mid \Gamma , \Delta \vdash C}{
      {-} \mid \Gamma \vdash A
      &
      B \mid \Delta \vdash C
    }
    \qquad
    \infer[\lright]{S \mid \Gamma \vdash A \lolli B}{S \mid \Gamma , A \vdash B}
    \\
    \infer[\unitl]{\I \mid \Gamma \vdash C}{{-} \mid \Gamma \vdash C}
    \qquad
    \infer[\tl]{A \ot B \mid \Gamma \vdash C}{A \mid B , \Gamma \vdash C}
    \qquad
    \infer[\unitr]{{-} \mid \quad \vdash \I}{}
    \qquad
    \infer[\tr]{S \mid \Gamma , \Delta \vdash A \ot B}{
      S \mid \Gamma \vdash A
      &
      {-} \mid \Delta \vdash B
    }
  \end{array}
\end{equation}

The inference rules in (\ref{eq:seqcalc}) are reminiscent of the ones in the sequent calculus for \NMILL\ \cite{abrusci:noncommutative:1990}, but there are some crucial differences.
\begin{enumerate}
\item The left logical rules $\unitl$, $\tl$ and $\lleft$, read bottom-up, are only allowed to be applied on the formula in the stoup position. In particular, there is no general way to remove a unit $\I$ nor decompose a tensor $A \ot B$ if these formulae are located in the context and not in the stoup (we will see in (\ref{eq:lleft:gen}) that something can actually be done to deal with implications $A \lolli B$ in the context).
\item The right tensor rule $\tr$, read bottom-up, splits the antecedent of the conclusion between the two premises whereby the formula in the stoup, in case such a formula is present, has to be moved to the stoup of the first premise. In particular, the stoup formula of the conclusion cannot be moved to the antecedent of the second premise even if $\Gamma$ is chosen to be empty.
\item The presence of the stoup implies a distinction between antecedents of forms $A \mid \Gamma$ and ${-} \mid A, \Gamma$. The structural rule $\pass$ (for `passivation'), read bottom-up, allows the moving of the leftmost formula in the context to the stoup position whenever the stoup is initially empty.
\item The logical connectives of \NMILL\ typically include two ordered implications $\lolli$ and $\illol$, which are two variants of linear implication arising from the removal of the exchange rule from intuitionistic linear logic. In \SkNMILL\ only one of the ordered implications (the left implication $\lolli$) is present. It is currently not clear to us whether the inclusion of the second implication to our logic is a meaningful addition and whether it corresponds to some particular categorical notion.
\end{enumerate}
The restrictions in 1--4 are essential for precisely capturing all the features of skew monoidal closed categories and nothing more, as we discuss in Section \ref{sec:catsem}.
Notice also that, similarly to the case of \NMILL, all structural rules of weakening, contraction and exchange are absent. We give names to derivations and we write $f : S \mid \Gamma \vdash A$ when $f$ is a particular derivation of the sequent $S \mid \Gamma \vdash A$.

Examples of valid derivations in the sequent calculus, corresponding to the structural laws $\lambda$, $\rho$ and $\alpha$ of skew monoidal closed categories (see Definition \ref{def:skewcat}) are given below.
\begin{equation}\label{eq:lra}
\small
  \begin{array}{c@{\;\quad}cc}
  (\lambda) & (\rho) & (\alpha) \\
  \infer[\tl]{\I \ot A \mid \quad \vdash A}{
    \infer[\unitl]{\I \mid A \vdash A}{
      \infer[\pass]{{-} \mid A \vdash A}{
        \infer[\ax]{A \mid \quad \vdash A}{}
      }
    }
  }
  &
  \infer[\tr]{A \mid \quad \vdash A \ot \I}{
    \infer[\ax]{A \mid \quad \vdash A}{}
    &
    \infer[\unitr]{{-} \mid \quad \vdash \I}{}
  }
  &
  \infer[\tl]{(A \ot B) \ot C \mid \quad \vdash A \ot (B \ot C)}{
    \infer[\tl]{A \ot B \mid C \vdash A \ot (B \ot C)}{
      \infer[\tr]{A \mid B , C \vdash A \ot (B \ot C)}{
        \infer[\ax]{A \mid \quad \vdash A}{}
        &
        \infer[\pass]{{-} \mid B , C \vdash B \ot C}{
          \infer[\tr]{B \mid C \vdash B \ot C}{
            \infer[\ax]{B \mid \quad \vdash B}{}
            &
            \infer[\pass]{{-} \mid C \vdash C}{
              \infer[\ax]{C \mid \quad \vdash C}{}
            }
          }
        }
      }
    }
  }
  \end{array}
\end{equation}
Examples of non-derivable sequents include the ``inverses'' of the conclusions in (\ref{eq:lra}), obtained by swapping the stoup formula with the succedent formula.
More precisely, the three sequents $X \mid ~ \vdash \I \ot X$, $X \ot \I \mid ~ \vdash X$ and $X \ot (Y \ot Z) \mid ~ \vdash (X \ot Y) \ot Z$ do not have any derivation. All possible attempts of constructing a valid derivation for each of them end in failure.
\begin{displaymath}
\small
  \begin{array}{ccc}
  (\lambda^{-1}) & (\rho^{-1}) & (\alpha^{-1}) \\[6pt]
    \infer[\tr]{X \mid ~\vdash \I \ot X}{
      \deduce[??]{X \mid ~ \vdash \I}{
      }
      &
      \deduce[??]{{-} \mid ~ \vdash X}{
      }
    }
    &
    \infer[\tl]{X \ot \I \mid \quad \vdash X}{
      \deduce{X \mid \I \vdash X}{??}
    }
    &
    \infer[\tl]{X \ot (Y \ot Z) \mid ~\vdash (X \ot Y) \ot Z}{
      \deduce{X \mid Y \ot Z \vdash (X \ot Y) \ot Z}{??}
    } \\
    (\text{$\tr$ sends $X$ to 1st premise}) &
    (\text{$\unitl$ does not act on $\I$ in context}) &
    (\text{$\tl$ does not act on $\ot$ in context})
  \end{array}
\end{displaymath}
Analogously, the sequents $\I \lolli A \mid ~ \vdash A$ and $(A \ot B) \lolli C \mid ~ \vdash A \lolli (B \lolli C)$ are derivable, while generally their ``inverses'' are not. Also, a derivation of $A \mid ~ \vdash B$ always yields a derivation of $\I \mid ~ \vdash A \lolli B$, but there are $A$, $B$ such that $\I \mid ~ \vdash A \lolli B$ is derivable while  $A \mid ~ \vdash B$ is not (take, e.g., $A = X$, $B = \I \ot X$).

Sets of derivations are quotiented by a congruence relation $\circeq$, generated by the following pairs of derivations.
\begin{equation}
\label{fig:circeq}
\begin{array}{rlll}
  \ax_{\I} &\circeq \unitl \text{ } (\unitr)
  \\
  \ax_{A \ot B} &\circeq \tl \text{ } (\tr \text{ } (\ax_{A} , \pass \text{ } \ax_{B}))
  \\
  \ax_{A \lolli B} &\circeq \lright \text{ } (\lleft \text{ } (\pass \text{ } \ax_{A}, \ax_{B} ))
  \\
  \tr \text{ } (\pass \text{ } f, g) &\circeq \pass \text{ } (\tr \text{ } (f, g)) &&f : A' \mid \Gamma \vdash A, g : {-} \mid \Delta \vdash B
  \\
  \tr \text{ } (\unitl \text{ } f, g) &\circeq \unitl \text{ } (\tr \text{ } (f , g)) &&f : {-} \mid \Gamma \vdash A , g : {-} \mid \Delta \vdash B
  \\
  \tr \text{ } (\tl \text{ } f, g) &\circeq \tl \text{ } (\tr \text{ } (f , g)) &&f : A' \mid B' , \Gamma \vdash A , g : {-} \mid \Delta \vdash B
  \\
  \tr \text{ } (\lleft \text{ } (f , g), h) & \circeq \lleft \text{ } (f, \tr \text{ } (g, h)) &&f: {-} \mid \Gamma \vdash A, g : B \mid \Delta \vdash C, h : {-} \mid \Lambda \vdash D
  \\
  \pass \text{ } (\lright \text{ } f) &\circeq \lright \text{ } (\pass \text{ } f) &&f : A' \mid \Gamma , A \vdash B
  \\
  \unitl \text{ } (\lright \text{ } f) &\circeq \lright \text{ } (\unitl \text{ } f) &&f : {-} \mid \Gamma , A \vdash B
  \\
  \tl \text{ } (\lright \text{ } f) &\circeq \lright \text{ } (\tl \text{ } f) &&f : A \mid B , \Gamma , C \vdash D
  \\
  \lleft \text{ } (f, \lright \text{ } g) &\circeq \lright \text{ } (\lleft \text{ } (f, g)) &&f : {-} \mid \Gamma \vdash A', g : B' \mid \Delta , A \vdash B
\end{array}
\end{equation}
The first three equations above are $\eta$-conversions, completely characterizing the $\ax$ rule on non-atomic formulae. The remaining equations are permutative conversions. The congruence $\circeq$ has been carefully chosen to serve as the proof-theoretic counterpart of the equational theory of skew monoidal closed categories, introduced in Definition \ref{def:skewcat}. The subsystem of equations involving only $(\I,\ot)$ originated in \cite{uustalu:sequent:2021} while the subsystem involving only $\lolli$ is from \cite{uustalu:deductive:nodate}.

\begin{theorem}
The sequent calculus enjoys cut admissibility: the following two cut rules are admissible.
  \begin{displaymath}
    \infer[\mathsf{scut}]{S \mid \Gamma , \Delta \vdash C}{
      S \mid \Gamma \vdash A
      &
      A \mid \Delta \vdash C
    }
    \qquad
    \infer[\mathsf{ccut}]{S \mid \Delta_0 , \Gamma , \Delta_1 \vdash C}{
      {-} \mid \Gamma \vdash A
      &
      S \mid \Delta_0 , A , \Delta_1 \vdash C
    }
  \end{displaymath}
\end{theorem}
The two cut rules satisfy a large number of
$\circeq$-equations, see e.g, \cite[Figures 5 and 6]{uustalu:sequent:2021} for the
list of such equations not involving $\lolli$.
In particular, the cut rules respect $\circeq$, in the sense that $\mathsf{scut}(f,g) \circeq \mathsf{scut}(f',g')$ whenever $f \circeq f'$ and $g \circeq g'$, and similarly for $\mathsf{ccut}$.

Here are some other interesting admissible rules relevant for the metatheory of this calculus.
\begin{itemize}
\item The left rules for $\I$ and $\ot$ are invertible up to $\circeq$, and similarly the right rule for $\lolli$.
No other rule is invertible; in particular, the
passivation rule $\pass$ is not.

\item
Applications of the invertible left logical rules can be iterated, and similarly for the invertible right $\lright$ rule, resulting in the two admissible rules
\begin{equation}\label{eq:inter:ante}
  \infer[\mathsf{L}^\star]{\ldbc S \mid \Gamma \rdbc_{\ot} \mid \Delta \vdash C}{
    S \mid \Gamma ,\Delta \vdash C
  }
  \qquad
  \infer[\lright^\star]{S \mid \Gamma \vdash \ldbc \Delta \mid C \rdbc_{\lolli}}{
    S \mid \Gamma , \Delta \vdash C
  }
\end{equation}
The interpretation of antecedents $\ldbc S \mid \Gamma \rdbc_{\ot}$ in (\ref{eq:inter:ante}) is the formula obtained by substituting the separator $\mid$ and the commas with tensors, $\ldbc S \mid A_1,\dots,A_n \rdbc_{\ot} = (\dots ((\ldbc S \llangle \ot A_1) \ot A_2) \dots ) \ot A_n$, where the interpretation of stoups is defined by $\ldbc {-} \llangle = \I$ and $\ldbc A \llangle = A$.
Dually, the formula $\ldbc \Delta \mid C \rdbc_{\lolli}$ in (\ref{eq:inter:ante}) is obtained by substituting $\mid$ and commas with implications:
$\ldbc A_1,\dots,A_n \mid C \rdbc_{\lolli} = A_1 \lolli (A_2 \lolli (\dots \lolli (A_n \lolli C)))$.

\item
Another left implication rule, acting on a formula $A \lolli B$ in the context, is derivable using cut:
\begin{equation}\label{eq:lleft:gen}
\small
    \!\!\!
  \proofbox{
    \infer[\lleft_{\mathsf{C}}]{S \mid \Delta_0, A \lolli B, \Gamma , \Delta_1 \vdash C}{
      \deduce{{-} \mid \Gamma \vdash A}{f}
      &
      \deduce{S \mid \Delta_0, B , \Delta_1 \vdash C}{g}
    }
  }
  {=}
  \proofbox{
    \infer[\mathsf{ccut}]{S \mid \Delta_0, A \lolli B, \Gamma , \Delta_1 \vdash C}{
      \infer[\pass]{{-} \mid A \lolli B, \Gamma \vdash B}{
        \infer[\lleft]{A \lolli B \mid \Gamma \vdash B}{
          \deduce{{-} \mid \Gamma \vdash A}{f}
          &
          \infer[\ax]{B \mid ~ \vdash B}{}
        }
      }
      &
      \deduce{S \mid \Delta_0, B , \Delta_1 \vdash C}{g}
    }
  }
\end{equation}
\end{itemize}

\paragraph{\SkNMILL\ as a Logic of Resources}

Similarly to other substructural logics like \MILL\ and \NMILL,\linebreak \SkNMILL\ can be understood
as a logic of resources. Under this perspective, formulae of the
sequent calculus in (\ref{eq:seqcalc}) correspond to
resources: atomic formulae are primitive resources; the formula $A \ot
B$ is read as ``resource $A$ before resource $B$''; $\I$ is ``nothing''; the formula $A \lolli B$
can be described as a method for turning the resource $A$ into the
resource $B$ (if $A \lolli B$ is provided before $A$). As in other substructural systems lacking the structural
rules of weakening and contraction, all resources are one-time
usable.

The antecedent of a sequent contains the resources at hand, while the
succedent contains the resource that needs to be produced. A
derivation is then a particular procedure for turning the available
resources into the goal resource. Under this interpretation,
derivations are naturally read and built from the conclusion to the
premises. Resources in the antecedent are ordered, meaning that they
need to be utilized in the order they appear. If a resource $A$
precedes another resource $B$ in the antecedent, then
$A$ must be consumed before $B$. The stoup position, when it is
non-empty, contains the resource that is immediately usable. The
resources in the context can only be spent after the resource in the stoup has been used. Time flows
bottom-up in proof trees, from conclusion to premises, and the proof
of a left premise always takes place before the proof of the right
premise.

The context shares similarities with the stack (and also the queue) data
structure. If, at some moment, there is no resource immediately consumable, as stoup is empty, then the
next available resource can be promoted to this status: the top (left-most position) of the context
can be popped and moved to the
stoup by the $\pass$ rule. A new resource can be pushed to the bottom (right-most
position) of the context using the $\lright$ rule. Using the $\tl$
rule, the immediately usable resource $A \ot B$ can be decomposed into two parts $A$ and $B$; the resource $A$ becomes immediately usable (remains in the stoup) while $B$ will become usable next (is pushed to the top of the context).  The resource $\I$, if immediately usable, can be disposed by the rule $\unitl$, making it possible to pop the next available
resource from the context. The resource $\I$ can always be produced
free-of-charge with the rule $\unitr$.

In the rule $\lleft$, we have a resource $A
\lolli B$ immediately usable and need to produce $C$. This gives us
access to the resource $B$, but only after we have produced $A$ by spending a part of the other resources available.  The context is split into two parts $\Gamma$ and $\Delta$. The first part $\Gamma$ is used to make $A$. Once this has been accomplished, production of $C$ can continue with $B$ immediately usable and $\Delta$ usable thereafter.

The succedent of the rule $\tr$ is of the form $A \ot B$, which
implies that first $A$ and then $B$ need to be produced. This justifies the
splitting of the context into two parts again: the first part $\Gamma$, consisting of
resources that can be spent sooner, is used to produce $A$ in the left
premise, while the second part $\Delta$ is spent subsequently for production of $B$ in the
right premise. Crucially, and this is one central ``skew'' aspect
of \SkNMILL, if we have a resource immediately usable in the
stoup $S$, this must be spent for the construction of the first resource
we are required to produce, namely $A$; it cannot be saved for producing $B$ even if producing $A$ needs no resources at all.

The other central ``skew'' aspect of \SkNMILL\ is that the left
rules only act on stoup formulae. This can be understood under the
resources-as-formulae correspondence to say that we are only allowed to
decompose an available resource when it has become immediately consumable (has entered the stoup). We are precluded from decomposing the
resources in the context ahead of their time. 

\section{Categorical Semantics via Skew Monoidal Closed Categories}
\label{sec:catsem}

Next we present a categorical semantics for the sequent calculus of \SkNMILL.

\begin{defn}\label{def:skewcat}
A \emph{(left) skew monoidal closed category} $\mathbb{C}$ is a category with a unit object $\I$ and two functors $\ot : \mathbb{C} \times \mathbb{C} \rightarrow \mathbb{C}$ and $\lolli : \mathbb{C}^{\mathsf{op}} \times \mathbb{C} \rightarrow \mathbb{C}$ forming an adjunction ${-} \ot B \dashv B \lolli {-}$ for all $B$,
and three natural transformations $\lambda$, $\rho$, $\alpha$ typed
  $\lambda_A : \I \ot A \to A$, $\rho_A : A \to A \ot \I$ and $\alpha_{A,B,C} : (A \ot B) \ot C \to A \ot (B \ot C)$,
satisfying the following equations due to Mac Lane \cite{maclane1963natural}:
\begin{center}
\begin{tikzcd}
	& {\I \ot \I} \\[-.2cm]
	\I && \I
	\arrow["{\rho_{\I}}", from=2-1, to=1-2]
	\arrow["{\lambda_{\I}}", from=1-2, to=2-3]
	\arrow[Rightarrow, no head, from=2-1, to=2-3]
\end{tikzcd}
\qquad
\begin{tikzcd}
	{(A \ot \I) \ot B} & {A \ot (\I \ot B)} \\[-.3cm]
	{A \ot B} & {A \ot B}
	\arrow[Rightarrow, no head, from=2-1, to=2-2]
	\arrow["{\rho_A \ot B}", from=2-1, to=1-1]
	\arrow["{A \ot \lambda_{B}}", from=1-2, to=2-2]
	\arrow["{\alpha_{A , \I , B}}", from=1-1, to=1-2]
\end{tikzcd}

\begin{tikzcd}
	{(\I \ot A ) \ot B} && {\I \ot (A \ot B)} \\[-.3cm]
	& {A \ot B}
	\arrow["{\alpha_{\I , A ,B}}", from=1-1, to=1-3]
	\arrow["{\lambda_{A \ot B}}", from=1-3, to=2-2]
	\arrow["{\lambda_{A} \ot B}"', from=1-1, to=2-2]
\end{tikzcd}
\qquad
\begin{tikzcd}
	{(A \ot B) \ot \I} && {A \ot (B \ot \I)} \\[-.3cm]
	& {A \ot B}
	\arrow["{\alpha_{A , B, \I}}", from=1-1, to=1-3]
	\arrow["{A \ot \rho_B}"', from=2-2, to=1-3]
	\arrow["{\rho_{A \ot B}}", from=2-2, to=1-1]
\end{tikzcd}

\begin{tikzcd}
	{(A\ot (B\ot C)) \ot D} && {A \ot ((B \ot C) \ot D)} \\[-.2cm]
	{((A \ot B)\ot C) \ot D} & {(A \ot B) \ot (C \ot D)} & {A \ot (B \ot (C \ot D))}
	\arrow["{\alpha_{A , B\ot C , D}}", from=1-1, to=1-3]
	\arrow["{A \ot \alpha_{B , C ,D}}", from=1-3, to=2-3]
	\arrow["{\alpha_{A ,B ,C\ot D}}"', from=2-2, to=2-3]
	\arrow["{\alpha_{A \ot B , C , D}}"', from=2-1, to=2-2]
	\arrow["{\alpha_{A , B ,C} \ot D}", from=2-1, to=1-1]
\end{tikzcd}
\end{center}
\end{defn}
The notion of skew monoidal closed category admits other equivalent characterizations \cite{street:skew-closed:2013,uustalu:eilenberg-kelly:2020}. Tuples of natural transformations $(\lambda , \rho , \alpha)$ are in bijective correspondence with tuples of (extra)natural transformations $(j, i, L)$ typed
$j_A : \I \to A \lolli A$, $i_A : \I \lolli A \to A$, $L_{A,B,C} : B \lolli C \to (A \lolli B) \lolli (A \lolli C)$.
Moreover, $\alpha$ and $L$ are interdefinable with a natural transformation $\mathsf{p}$ typed $\mathsf{p}_{A , B , C} : (A \ot B) \lolli C \to \linebreak A \lolli (B \lolli C)$, embodying an internal version of the adjunction between $\ot$ and $\lolli$.

\begin{example}[from \cite{uustalu:eilenberg-kelly:2020}]
This example explains how to turn every categorical model of \MILL\ extended with a $\Box$-like modality of necessity (or something like the exponential modality $!$ of linear logic) into a model of \SkNMILL.
Let $(\mathbb{C},\I,\ot,\lolli)$ be a (possibly symmetric) monoidal closed category and let $(D,\varepsilon, \delta)$ be a comonad on $\mathbb{C}$, where $\varepsilon_A : D\,A \to A$ and $\delta_A : D\,A \to D\,(D\,A)$ are the counit and comultiplication of $D$. Suppose the comonad $D$ to be \emph{lax monoidal}, i.e., coming with a map $\mathsf{e} : \I \to D\,I$ and a natural transformation $\mathsf{m}$ typed $\mathsf{m}_{A,B} : D \,A \ot D\,B \to D\,(A \ot B)$ cohering suitably with $\lambda$, $\rho$, $\alpha$, $\varepsilon$ and $\delta$.
Then $\mathbb{C}$ has also a skew monoidal closed structure $(\I, \otd, \lollid)$ given by  $A \otd B = A \ot D\,B$ and $B \lollid C = D\,B \lolli C$. The adjunction ${-} \ot D\,B \dashv D\,B \lolli {-}$ yields an adjunction ${-} \otd B \dashv B \lollid {-}$. The structural laws are
\[
\begin{array}{c}
\lambda^D_A \, = \, \xymatrix@C=3pc{\I \ot D\,A \ar[r]^-{\I \ot \varepsilon_A} & \I \ot A \ar[r]^-{\lambda_A} & A}
\hspace*{1.5cm}
\rho^D_A \, = \, \xymatrix@C=3pc{A \ar[r]^-{\rho_A} & A \ot \I \ar[r]^-{A \ot \mathsf{e}} & A \ot D\,\I}
\\
\begin{array}{rcl}
\alpha^D_{A,B,C} & = & \xymatrix@C=5pc{(A \ot D\,B) \ot D\,C
                   \ar[r]^-{(A \ot D B) \ot \delta_C}
                   & (A \ot D\,B) \ot D\,(D\,C)} \\
& & \hspace*{1.5cm} \xymatrix@C=5pc{\ar[r]^-{\alpha_{A,DB,D(DC)}}
                  & A \ot (D\,B \ot D\,(D\,C)) \ar[r]^-{A \ot \mathsf{m}_{B,DC}}
                  & A \ot D\,(B \ot D\,C)}
\end{array}
\end{array}
\]
$(\mathbb{C},\I,\otd,\lollid)$  is a ``genuine'' skew monoidal closed category, in the sense that $\lambda^D$, $\rho^D$ and $\alpha^D$ are all generally non-invertible.
\end{example}

\begin{defn}
  A \emph{(strict) skew monoidal closed functor} $F : \mathbb{C} \rightarrow \mathbb{D}$ between skew monoidal closed categories $(\mathbb{C} , \I , \ot , \lolli)$ and $(\mathbb{D} , \I' , \ot' , \lolli')$ is a functor from $\mathbb{C}$ to $\mathbb{D}$ satisfying
    $F \I = \I'$, $F (A \ot B) = \linebreak F A \ot' F B$ and
    $F(A \lolli B) = F A \lolli' F B$, also preserving the structural laws $\lambda$, $\rho$ and $\alpha$ on the nose.
\end{defn}

The formulae, derivations and the equivalence relation $\circeq$ of the sequent calculus for \SkNMILL\ determine a skew monoidal closed category $\FSkMCC(\mathsf{At}$).
\begin{defn}\label{def:fskmcc}
  The skew monoidal closed category $\FSkMCC(\mathsf{At})$ has
  as objects formulae; the operations $\I$, $\ot$ and $\lolli$ are the logical connectives. The set of maps between objects $A$ and $B$ is the set of derivations $A \mid ~ \vdash B$ quotiented by the equivalence relation $\circeq$. The identity map on $A$ is the equivalence class of $\ax_A$, while composition is given by $\mathsf{scut}$. The structural laws $\lambda$, $\rho$, $\alpha$ are given by derivations in (\ref{eq:lra}). 
\end{defn}
 This is a good definition since all equations of a skew monoidal closed category turn out to hold.

Skew monoidal closed categories with given interpretations of atoms into them constitute models of the sequent calculus of \SkNMILL, in the sense specified by the following theorem.
\begin{theorem}\label{thm:models}
  Let $\mathbb{D}$ be a skew monoidal closed category. Given $F_{\mathsf{At}} : \mathsf{At} \rightarrow |\mathbb{D}|$ providing evaluation of atomic formulae as objects of $\mathbb{D}$, there exists a skew monoidal closed functor $F : \FSkMCC(\mathsf{At}) \rightarrow \mathbb{D}$.
\end{theorem}
\begin{proof}
  Let $(\mathbb{D} , \I' , \ot' , \lolli')$ be a skew monoidal closed category.
  The action on object $F_0$ of the functor $F$ is defined by induction on the input formula:
  \begin{equation*}
    F_0X = F_{\mathsf{At}}X
    \qquad
    F_0\I = \I'
    \qquad
    F_0(A \ot B) = F_0A \ot' F_0B
    \qquad
    F_0(A \lolli B) = F_0A \lolli' F_0B
  \end{equation*}
  The encoding of antecedents as formulae $\ldbc S \mid \Gamma \rdbc_{\ot}$, introduced immediately after (\ref{eq:inter:ante}), can be replicated also in $\mathbb{D}$ by simply replacing $\I$ and $\ot$ with $\I'$ and $\ot'$ in the definition, where now $S$ is an optional object and $\Gamma$ is a list of objects of $\mathbb{D}$. Using this encoding, it is possible to show that each rule in (\ref{eq:seqcalc}) is derivable in $\mathbb{D}$. As an illustrative case, consider the rule $\pass$. Assume given a map $f : \ldbc A \mid \Gamma \rdbc_{\ot'} \to C$ in $\mathbb{D}$. Then, assuming $\Gamma = A_1,\dots,A_n$, we can define the passivation of $f$ typed $\ldbc {-} \mid A, \Gamma \rdbc_{\ot'} \to C$ as
\[\xymatrixcolsep{7pc}
\xymatrix{
  (\dots ((\I' \ot' A_1) \ot' A_2) \dots ) \ot' A_n
  \ar[r]^-{(\dots (\lambda'_{A_1} \ot' A_2) \dots ) \ot' A_n} &
  (\dots (A_1 \ot' A_2) \dots ) \ot' A_n
  \ar[r]^-{f} &
  C
}
\]
This implies the existence of a function $F_1$, sending each derivation $f : S \mid \Gamma \vdash A$ to a map $F_1f : \linebreak F_0(\ldbc S \mid \Gamma \rdbc_{\ot}) \to FA$ in $\mathbb{D}$, defined by induction on the derivation $f$. When restricted to sequents of the form $A \mid ~ \vdash B$, the function $F_1$ provides the action on maps of $F$.
It is possible to show that $F$ preserves the skew monoidal closed structure, so it is a skew monoidal closed functor.
\end{proof}

Moreover, the sequent calculus as a presentation of a skew monoidal closed category is the \emph{initial} one among these
models, or, equivalently, $\FSkMCC(\mathsf{At})$ is the \emph{free}
such category on the set $\mathsf{At}$.

\begin{theorem}\label{thm:unique}
  The skew monoidal closed functor $F$ constructed in Theorem \ref{thm:models} 
  is the unique one for which $F_0 X = F_{\mathsf{At}} X$ for any atom $X$.
\end{theorem}

Alternatively, following the strategy of our previous papers \cite{uustalu:sequent:2021,uustalu:proof:nodate,uustalu:deductive:nodate,veltri:coherence:2021}, this result can be shown by going via a Hilbert-style deductive system that directly presents the free skew monoidal closed category on $\mathsf{At}$. Its formulae are the same of the sequent calculus, its sequents are pairs $A \Rightarrow B$, with $A$ and $B$ single formulae. Derivations are generated by the following rules, in which one can  recognize all the structure of Definition \ref{def:skewcat} (in particular $\pi$ and $\pi^{-1}$ are the adjunction data).
\begin{displaymath}
  \def\arraystretch{1.5}
  \begin{array}{c}
        \infer[\id]{A \Rightarrow A}{}
        \quad
        \infer[\mathsf{comp}]{A \Rightarrow C}{
          A \Rightarrow B
          &
          B \Rightarrow C
        }
        \qquad
      \infer[\otimes]{A \ot B \Rightarrow C \ot D}{
        A \Rightarrow C
        &
        B \Rightarrow D
      }
      \qquad
      \infer[\lolli]{A \lolli B \Rightarrow C \lolli D}{
        C \Rightarrow A
        &
        B \Rightarrow D
      }
      \\
      \infer[\lambda]{\I \ot A \Rightarrow A}{}
      \quad
      \infer[\rho]{A \Rightarrow A \ot \I}{}
      \quad
      \infer[\alpha]{(A \ot B) \ot C \Rightarrow A \ot (B \ot C)}{}
      \qquad
      \infer[\pi]{A \Rightarrow B \lolli C}{A \ot B \Rightarrow C}
      \quad
      \infer[\pi^{-1}]{A \ot B \Rightarrow C}{A \Rightarrow B \lolli C}
  \end{array}
\end{displaymath}
There exists a congruence relation $\doteq$ on each set of derivations $A \Rightarrow B$ with generators directly encoding the equations of the theory of skew monoidal closed categories. With the Hilbert-style calculus in place, one can show that there is a bijection between the set of derivations of the sequent $A \mid ~ \vdash B$ modulo $\circeq$ and the set of derivations of the sequent $A \Rightarrow B$ modulo $\doteq$.

\section{Proof-Theoretic Semantics via Focusing}
\label{sec:focus}

The equivalence relation $\circeq$ from (\ref{fig:circeq}) can also be viewed as an abstract rewrite system, by orienting every equation from left to right. The resulting rewrite system is locally confluent and strongly normalizing, thus confluent with unique normal forms. Derivations in normal form thus correspond to canonical representatives of $\circeq$-equivalence classes.
These representatives can be organized in a \emph{focused sequent calculus} in the sense of \cite{andreoli:logic:1992}, which describes, in a declarative fashion, a root-first proof search strategy for the (original, unfocused) sequent calculus.

\subsection{A First (Na{\"i}ve) Focused Sequent Calculus}
As a first attempt to focusing, we na{\"i}vely merge together the rules of the focused sequent calculi of skew monoidal categories \cite{uustalu:sequent:2021} and skew prounital closed categories \cite{uustalu:deductive:nodate}. In the resulting calculus, sequents have one of 4 possible subscript annotations, corresponding to 4 different phases of proof search: $\RI$ for `right invertible', $\LI$ for `left invertible`, $\Pass$ for `passivation` and $\F$ for `focusing`. We will see soon that this focused sequent calculus is too permissive, in the sense that two distinct derivations in the focused system can correspond to $\circeq$-equivalent sequent calculus derivations.

\begin{equation}\label{eq:naive}
  \begin{array}{lc}
    \text{(right invertible)} & 
    \proofbox{
    \infer[\lright]{S \mid \Gamma \vdash_{\RI} A \lolli B}{S \mid \Gamma , A \vdash_{\RI} B}
    \qquad
    \infer[\LI 2 \RI]{S \mid \Gamma \vdash_{\RI} P}{S \mid \Gamma \vdash_{\LI} P}
    }
    \\[9pt]
    \text{(left invertible)} & 
    \proofbox{
    \infer[\unitl]{\I \mid \Gamma \vdash_{\LI} P}{{-} \mid \Gamma \vdash_{\LI} P}
    \qquad
    \infer[\tl]{A \ot B \mid \Gamma \vdash_{\LI} P}{A \mid B , \Gamma \vdash_{\LI} P}
    \qquad
    \infer[\Pass 2 \LI]{T \mid \Gamma \vdash_{\LI} P}{T \mid \Gamma \vdash_{\Pass} P}
    }
    \\[9pt]
    \text{(passivation)} & 
    \proofbox{
      \infer[\pass]{{-} \mid A , \Gamma \vdash_{\Pass} P }{
      A\mid \Gamma \vdash_{\LI} P
    }
    \qquad
    \infer[\F 2 \Pass]{T \mid \Gamma \vdash_{\Pass} P}{
      T \mid \Gamma \vdash_{\F} P
    }
    }
    \\
    \text{(focusing)} &    \\ 
    \multicolumn{2}{c}{
    \hspace*{3mm}
    \infer[\ax]{X \mid \quad \vdash_{\F} X}{}
    \quad
    \infer[\unitr]{{-} \mid \quad \vdash_{\F} \I}{}
    \quad
    \infer[\tr]{T \mid \Gamma , \Delta \vdash_{\F} A \ot B}{
      T \mid \Gamma \vdash_{\RI} A
      &
      {-} \mid \Delta \vdash_{\RI} B
    }
    \quad
    \infer[\lleft]{A \lolli B \mid \Gamma , \Delta \vdash_{\F} P}{
      {-} \mid \Gamma \vdash_{\RI} A
      &
      B \mid \Delta \vdash_{\LI} P
    }
    }
  \end{array}
\end{equation}
In the rules above and the rest of the paper, the metavariable $P$ denotes a \emph{positive} formula, i.e. $P \not= A \lolli B$, while metavariable $T$ indicates a \emph{negative} stoup, i.e. $T \not= \I$ and $T\not= A \ot B$ ($T$ can also be empty).

We explain the rules of the focused sequent calculus from the perspective of root-first proof search. The starting phase is `right invertible` $\RI$. 
\begin{itemize}
\item[($\vdash_\RI$)] We repeatedly apply the right invertible rule $\lright$ with the goal of reducing the succedent to a positive formula $P$.
  When the succedent formula becomes positive, we move to phase $\LI$ via $\LI2\RI$.
\item[($\vdash_\LI$)] We repeatedly destruct the stoup formula via application of left invertible rules $\tl$ and $\unitl$ with the goal of making it negative. When this happens, we move to phase $\Pass$ via $\Pass2\LI$.
\item[($\vdash_\Pass$)] We have the possibility of applying the passivation rule and move the leftmost formula $A$ in the context to the stoup when the latter is empty. This allows us to start decomposing $A$ using left invertible rules in phase $\LI$. Otherwise, we move to phase $\F$ via $\F2\Pass$.
\item[($\vdash_\F$)] We apply one of the four remaining rules $\ax$, $\unitr$, $\tr$ or $\lleft$. The premises of $\tr$ are both in phase $\RI$ since $A$ and $B$ are generic formulae, in particular they could be implications. The first premise of $\lleft$ is in phase $\RI$ for the same reason while the second premise is in $\LI$ because the succedent formula $P$ is positive.
\end{itemize}

The focused calculus in (\ref{eq:naive}) is sound and complete wrt.\ the sequent calculus in (\ref{eq:seqcalc}) in regards to derivability, but not \emph{equationally complete}, i.e., there exist $\circeq$-equivalent sequent calculus derivations which have multiple distinct derivations using the rules in (\ref{eq:naive}). In other words, the rules in (\ref{eq:naive}) are too permissive. They facilitate two forms of non-determinism in root-first proof search that should not be there.

\begin{enumerate}

\item[(i)] The first premise of the $\tr$ rule is in phase $\RI$, since $A$ is potentially an implication which the invertible right rule $\lright$ could act upon. Proof search for the first premise eventually hits phase $\Pass$, when we have the possibility of applying the $\pass$ rule if the stoup is empty. This implies the existence of situations where either of the rules $\tr$ and $\pass$ can be applied first, in both cases resulting in valid focused derivations.
  As an example, consider the two distinct derivations of ${-} \mid X , \Gamma , \Delta \vdash_{\Pass} P \ot C$ under assumptions $f : X \mid \Gamma \vdash_{\LI} P$ and $g : {-} \mid \Delta \vdash_{\RI} C$.
  \begin{equation}\label{eq:passtr}
  \small
    \proofbox{
      \infer[\pass]{{-} \mid X , \Gamma , \Delta \vdash_{\Pass} P \ot C}{
        \infer[\mathsf{sw}]{X \mid \Gamma , \Delta \vdash_{\LI} P \ot C}{
          \infer[\tr]{X \mid \Gamma , \Delta \vdash_{\F} P \ot C}{
            \infer[\mathsf{sw}]{X \mid \Gamma \vdash_{\RI} P}{
              \deduce{X \mid \Gamma \vdash_{\LI} P}{f}
            }
            &
            \deduce{{-} \mid \Delta \vdash_{\RI} C}{g}
          }
        }
      }
    \qquad
    \infer[\mathsf{sw}]{{-} \mid X , \Gamma , \Delta \vdash_{\Pass} P \ot C}{
      \infer[\tr]{{-} \mid X , \Gamma , \Delta \vdash_{\F} P \ot C}{
        \infer[\mathsf{sw}]{{-} \mid X , \Gamma \vdash_{\RI} P}{
          \infer[\pass]{{-} \mid X , \Gamma \vdash_{\Pass} P}{
            \deduce{X \mid \Gamma \vdash_{\LI} P}{f}
          }
        }
        &
        \deduce{{-} \mid \Delta \vdash_{\RI} C}{g}
      }
    }
    }
  \end{equation}
  Here and in the rest of the paper the rule $\mathsf{sw}$ above stands for a sequence of (appropriately typed) phase
  switching inferences by $\LI2\RI$, $\Pass2\LI$ and $\F2\Pass$.
  The corresponding sequent calculus derivations are equated by congruence relation $\circeq$ because of the 4th equation from (\ref{fig:circeq}), i.e., the permutative conversion involving $\tr$ and $\pass$.

  \item[(ii)] Rules $\tr$ and $\lleft$ appear in the same phase $\F$, though there are situations where both rules can be applied first, which can lead to two distinct focused derivations. More precisely, there are cases when $\tr$ and $\lleft$ can be interchangeably applied. As an example, consider the following two valid derivations of $A \lolli X \mid \Gamma , \Delta , \Lambda \vdash_{\F} P \ot D$ under the assumption of $f : {-} \mid \Gamma \vdash_{\RI} A$, $g : X \mid \Delta \vdash_{\LI} P$ and $h : {-} \mid \Lambda \vdash_{\RI} D$.
    \begin{equation}\label{eq:llefttr}
\small
      \hspace{-.5cm}
      \proofbox{
      \infer[\lleft]{A \lolli X \mid \Gamma , \Delta , \Lambda \vdash_{\F} P \ot D}{
        \deduce{{-} \mid \Gamma \vdash_{\RI} A}{f}
        &\hspace{-.6cm}
        \infer[\mathsf{sw}]{X \mid \Delta , \Lambda \vdash_{\LI} P \ot D}{
          \infer[\tr]{X \mid \Delta , \Lambda \vdash_{\F} P\ot D}{
            \infer[\mathsf{sw}]{X \mid \Delta \vdash_{\RI} P}{
              \deduce{X \mid \Delta \vdash_{\LI} P}{g}
            }
            &
            \deduce{{-} \mid \Lambda \vdash_{\RI} D}{h}
          }
        }
      }
      \hspace{.2cm}
      \infer[\tr]{A \lolli X \mid \Gamma , \Delta , \Lambda \vdash_{\F} P \ot D}{
        \infer[\mathsf{sw}]{A \lolli X \mid \Gamma , \Delta \vdash_{\RI} P}{
          \infer[\lleft]{A \lolli X \mid \Gamma , \Delta \vdash_{\F} P}{
            \deduce{{-} \mid \Gamma \vdash_{\RI} A}{f}
            &
            \deduce{X \mid \Delta \vdash_{\LI} P}{g}
          }
        }
        &\hspace{-.6cm}
        \deduce{{-} \mid \Lambda \vdash_{\RI} D}{h}
      }
    }
    \end{equation}
    The corresponding sequent calculus derivations, at the same time, are $\circeq$-equivalent because of the 7th equation from (\ref{fig:circeq}), the permutative conversion for $\tr$ and $\lleft$.
\end{enumerate}

To get rid of type (i) undesired non-determinism, one might try an idea similar to the one that works in the skew monoidal non-closed case \cite{uustalu:sequent:2021}, namely, to prioritize $\pass$ over $\tr$ by requiring the first premise of the latter to be a sequent in phase $\F$.
But this does not do the right thing in the skew monoidal closed case. E.g., the sequent ${-} \mid Y \vdash_{\F} (X \lolli X) \ot Y$ becomes underivable while its counterpart is derivable in (\ref{eq:seqcalc}).
\begin{equation*}\label{eq:counterexample0}
  \infer[\tr]{{-} \mid Y \vdash_{\F} (X \lolli X) \ot Y}{
    \deduce{{-} \mid \quad \vdash_{\F} X \lolli X}{??}
    &
    \infer[\mathsf{sw}]{{-} \mid Y \vdash_{\RI} Y}{
      \infer[\pass]{{-} \mid Y \vdash_{\Pass} Y}{
        \infer[\mathsf{sw}]{Y \mid \quad \vdash_{\LI} Y}{
          \infer[\ax]{Y \mid \quad \vdash_{\F} Y}{}
        }
      }
    }
  }
\end{equation*}

An impulsive idea for eliminating undesired non-determinism of type (ii) is to prioritize the application of $\lleft$ over $\tr$, e.g., by forcing the application of $\lleft$ in phase $\F$ whenever the stoup formula is an implication and restricting the application of $\tr$ to sequents where the stoup is empty or atomic. This too leads to an unsound calculus, since the sequent $X \lolli Y \mid Z \vdash_{\F} (X \lolli Y) \ot Z$, which has a derivable correspondent in (\ref{eq:seqcalc}), would not be derivable by first applying the $\lleft$ rule.
\begin{equation*}\label{eq:counterexample1}
   \infer[\lleft]{X \lolli Y \mid Z \vdash_{\F} (X \lolli Y) \ot Z}{
     \deduce{{-} \mid \quad \vdash_{\RI} X}{??}
     &
     \infer[\mathsf{sw}]{Y \mid Z \vdash_{\LI} (X \lolli Y) \ot Z}{
       \infer[\tr]{Y \mid Z \vdash_{\F} (X \lolli Y) \ot Z}{
         \infer[\lright]{Y \mid \quad \vdash_{\RI} X \lolli Y}{
           \deduce{Y \mid X \vdash_{\RI} Y}{??}
         }
         &
         \infer[\mathsf{sw}]{{-} \mid Z \vdash_{\RI} Z}{
           \infer[\pass]{{-} \mid Z \vdash_{\Pass} Z}{
             \infer[\mathsf{sw}]{Z \mid \quad \vdash_{\LI} Z}{
               \infer[\ax]{Z \mid \quad \vdash_{\F} Z}{}
             }
           }
         }
       }
     }
   }
\end{equation*}
Dually, prioritizing the application of $\tr$ over $\lleft$ leads to similar issues, e.g., the sequent $X \lolli (Y \ot Z) \mid X \vdash_\F Y \ot Z$ would not be derivable by first applying the $\tr$ rule while its counterpart is derivable in  (\ref{eq:seqcalc}).
\begin{equation*}\label{eq:counterexample2}
\infer[\tr]{X \lolli (Y \ot Z) \mid X \vdash_{\F} Y \ot Z}{
  \infer[\mathsf{sw}]{X \lolli (Y \ot Z) \mid X \vdash_{\RI} Y}{
    \infer[\lleft]{X \lolli (Y \ot Z) \mid X \vdash_{\F} Y}{
      \infer[\mathsf{sw}]{{-}  \mid X \vdash_{\RI} X}{
        \infer[\pass]{{-}  \mid X \vdash_{\Pass} X}{
          \infer[\mathsf{sw}]{X \mid \quad \vdash_{\LI} X}{
            \infer[\ax]{X \mid \quad \vdash_{\F} X}{}
          }
        }
      }
      &
      \infer[\tl]{Y \ot Z \mid \quad \vdash_{\LI} Y}{
        \deduce{Y \mid Z \vdash_{\LI} Y}{??}
      }
    }
  }
  &
  \deduce{{-}  \mid \quad \vdash_{\RI} Z}{??}
}
\end{equation*}

\subsection{A Focused System with Tag Annotations}\label{sec:tag}

In order to eliminate undesired non-determinism of type (i) between $\pass$ and $\tr$, we need to restrict applications of $\pass$ in the derivation of the first premise of an application of $\tr$.
One way to achieve this is to force that such an application of $\pass$ is allowed only if the leftmost formula of the context is \emph{new}, in the sense that it was not already present in the context before the $\tr$ application. For example, with this restriction in place, the application of $\pass$ in the 2nd derivation of (\ref{eq:passtr}) would be invalid, since the formula $X$ was already present in context before the application of $\tr$.

Analogously, undesired non-determinism of type (ii) between $\lleft$ and $\tr$ can be eliminated by restricting applications of $\lleft$ after an application of $\tr$. This can be achieved by forcing the subsequent application of $\lleft$ to split the context into two parts $\Gamma,\Delta$ in such a way that $\Gamma$, i.e., the context of the first premise, necessarily contains some \emph{new} formula occurrences that were not in the context before the first $\tr$ application. Under this restriction, the application of $\lleft$ in the 2nd derivation of (\ref{eq:llefttr}) would become invalid, since all formulae in $\Gamma$ are already present in context before the application of $\tr$.

One way to distinguish between old and new formulae occurrences in the above cases is to mark with a \emph{tag} $\bullet$ each new formula appearing in context during the building of a focused derivation. We christen a formula occurrence ``new'' whenever it is moved from the succedent to the context via an application of the right implication rule $\lright$. In order to remember when we are building a derivation of a sequent arising as the first premise of $\tr$, in which the distinction between old and new formula is relevant, we mark such sequents with a tag $\bullet$ as well.
More generally, we write $S \mid \Gamma \xvdash_{ph} C$ for a sequent that can be untagged or tagged, i.e., the turnstile can be of the form $\vdash_{ph}$ or $\vdash^\bullet_{ph}$, for $ph \in \{ \RI,\LI,\Pass,\F\}$. This implies that there are a total of eight sequent phases, corresponding to the possible combinations of four subscript phases with the untagged/tagged state.
In tagged sequents  $S \mid \Gamma \vdash_{ph}^{\bullet} C$, the formulae in the context $\Gamma$ can be untagged or tagged, i.e., they can be of the form $A$ or $A^\bullet$; to be precise, all untagged formulae in $\Gamma$ must precede all tagged formulae (i.e., the context splits into untagged and tagged parts and, instead of possibly tagged formulae, we could alternatively work with contexts with two compartments). The formulae in the context of an untagged sequent $S \mid \Gamma \vdash_{ph} C$ must all be untagged (or, alternatively, the tagged compartment must be empty). Given a context $\Gamma$, we write $\Gamma^{\circ}$ for the same context where all tags have been removed from the formulae in it.

Derivations in the focused sequent calculus with tag annotations are generated by the rules
\begin{equation}\label{eq:focus}
  \begin{array}{lc}
    \text{(right invertible)} & 
    \proofbox{
      \infer[\lright]{S \mid \Gamma \vdash^{x}_{\RI} A \lolli B}{S \mid \Gamma , A^{x} \vdash^{x}_{\RI} B}
    \qquad
    \infer[\LI 2 \RI]{S \mid \Gamma \vdash^{x}_{\RI} P}{S \mid \Gamma \vdash^{x}_{\LI} P}
    }
    \\[10pt]
    \text{(left invertible)} & 
    \proofbox{
      \infer[\unitl]{\I \mid \Gamma \vdash_{\LI} P}{{-} \mid \Gamma \vdash_{\LI} P}
    \qquad
    \infer[\tl]{A \ot B \mid \Gamma \vdash_{\LI} P}{A \mid B , \Gamma \vdash_{\LI} P}
    \qquad
    \infer[\Pass 2 \LI]{T \mid \Gamma \xvdash_{\LI} P}{T \mid \Gamma \xvdash_{\Pass} P}
    }
    \\[10pt]
    \text{(passivation)} &
    \proofbox{
    \infer[\pass]{{-} \mid A^{x} , \Gamma \xvdash_{\Pass} P }{
      A\mid \Gamma^{\circ} \vdash_{\LI} P
    }
    \qquad
    \infer[\F 2 \Pass]{T \mid \Gamma \xvdash_{\Pass} P}{
      T \mid \Gamma \xvdash_{\F} P 
    }
    }
    \\[10pt]
    \text{(focusing)} &    
    \proofbox{\infer[\ax]{X \mid \quad \xvdash_{\F} X}{}
    \qquad
    \infer[\unitr]{{-} \mid \quad \xvdash_{\F} \I}{}
    }
    \\[6pt]
    \multicolumn{2}{c}{
    \infer[\tr]{T \mid \Gamma , \Delta \xvdash_{\F} A \ot B}{
      T \mid \Gamma^{\circ} \vdash^{\bullet}_{\RI} A
      &
      {-} \mid \Delta^{\circ} \vdash_{\RI} B
    }
    \qquad
    \infer[\lleft]{A \lolli B \mid \Gamma , \Delta \xvdash_{\F} P}{
      {-} \mid \Gamma^{\circ} \vdash_{\RI} A
      &
      B \mid \Delta^{\circ} \vdash_{\LI} P
      &
      x = \bullet \supset \bullet \in \Gamma
    }
    }
  \end{array}
\end{equation}
Remember that $P$ is a positive formula and $T$ is a negative stoup. The side condition in rule $\lleft$ reads: if $x = \bullet$, then some formula in $\Gamma$ must be tagged. For the rule $\pass$ notice that, if $x = \bullet$, it is actually forced that all formulae of $\Gamma$ are tagged since the preceding context formula $A^\bullet$ is tagged. For the rules $\tr$ and $\lleft$ similarly notice that, 
if some formula of $\Gamma$ is tagged, then all formulae of $\Delta$ must be tagged. 

The rules in (\ref{eq:focus}), when stripped of all the tags, are equivalent to the rules in the na{\"i}ve calculus (\ref{eq:naive}). When building a derivation of an untagged sequent $S \mid \Gamma \vdash_\RI A$, the only possible way to enter a tagged phase is via an application of the $\tr$ rule, so that sequents with turnstile marked $\vdash_{ph}^\bullet$ denote the fact that we are performing proof search for the first premise of an $\tr$ inference (and the stoup is negative).
The search for a proof of a tagged sequent $T \mid \Gamma^\circ \vdash^\bullet_\RI A$ proceeds as follows:
\begin{itemize}
\item[($\vdash^\bullet_\RI$)] We eagerly apply the right invertible rule $\lright$ with the goal of reducing $A$ to a positive formula $P$. All formulae that get moved to the right end of the context are ``new'', and are therefore marked with $\bullet$.
  When the succedent formula becomes positive, we move to the tagged $\LI$ phase via $\LI2\RI$.
\item[($\vdash^\bullet_\LI$)] Since $T$ is a negative stoup, we can only move to the tagged $\Pass$ phase via $\Pass2\LI$.
\item[($\vdash^\bullet_\Pass$)] If the stoup is empty, we have the possibility of applying the $\pass$ rule and move the leftmost formula $A$ in the context to the stoup, but only when this formula is marked by $\bullet$. This restriction makes it possible to remove undesired non-determinism of type (i). We then strip the context of all tags and jump to the untagged $\LI$ phase. If we do not (or cannot) apply $\pass$, we move to the tagged $\F$ phase via $\F2\Pass$.
\item[($\vdash^\bullet_\F$)] The possible rules to apply (depending on the stoup and succedent formula) are $\ax$, $\unitr$, $\tr$ or $\lleft$. If we apply $\tr$, we remove all tags from the context $\Gamma, \Delta$ and move the first premise to the tagged $\RI$ phase again. The tags are removed from the context in order to reset tracking of new formulae. The most interesting case is $\lleft$, which can only be applied if the $\Gamma$ part of the context $\Gamma, \Delta$ contains at least one tagged formula. This side condition implements the restriction allowing the elimination of undesired non-determinism of type (ii). All tags are removed from $\Gamma, \Delta$ and proof search continues in the appropriate untagged phases.
\end{itemize}

The employment of tag annotations eliminates the two types of undesired non-determinism. For example, only one of the two derivations in (\ref{eq:passtr}) is valid using the rules in (\ref{eq:focus}), and similarly for (\ref{eq:llefttr}).
\begin{displaymath}
\small
  \begin{array}{cc}
      \infer[\pass]{{-} \mid X , \Gamma , \Delta \vdash_{\Pass} P \ot C}{
        \infer[\mathsf{sw}]{X \mid \Gamma , \Delta \vdash_{\LI} P \ot C}{
          \infer[\tr]{X \mid \Gamma , \Delta \vdash_{\F} P \ot C}{
            \infer{X \mid \Gamma \vdash^{\bullet}_{\RI} P}{
              \deduce{X \mid \Gamma \vdash^{\bullet}_{\LI} P}{f}
              }
            &
            \deduce{{-} \mid \Delta \vdash_{\RI} C}{g}
          }
        }
      }
    &
    \infer[\mathsf{sw}]{{-} \mid X , \Gamma , \Delta \vdash_{\Pass} P \ot C}{
      \infer[\tr]{{-} \mid X , \Gamma , \Delta \vdash_{\F} P \ot C}{
        \infer[\mathsf{sw}]{{-} \mid X , \Gamma \vdash^{\bullet}_{\RI} P}{
          \deduce{{-} \mid X , \Gamma \vdash^{\bullet}_{\Pass} P}{??}
        }
        &
        \deduce{{-} \mid \Delta \vdash_{\RI} C}{g}
      }
    }
    \\
    (\text{same derivation as in (\ref{eq:passtr})})
    &
    (\text{$\pass$ not applicable since $X$ is not tagged})
    \\[10pt]
      \infer[\lleft]{A \lolli X \mid \Gamma , \Delta , \Lambda \vdash_{\F} P \ot D}{
        \deduce{{-} \mid \Gamma \vdash_{\RI} A}{f}
        &
        \infer[\mathsf{sw}]{X \mid \Delta , \Lambda \vdash_{\LI} P \ot D}{
          \infer[\tr]{X \mid \Delta , \Lambda \vdash_{\F} P \ot D}{
            \infer[\mathsf{sw}]{X \mid \Delta \vdash^{\bullet}_{\RI} P}{
              \deduce{X \mid \Delta \vdash^{\bullet}_{\LI} P}{g}
              }
            &
            \deduce{{-} \mid \Lambda \vdash_{\RI} D}{h}
          }
        }
      }
    &
      \infer[\tr]{A \lolli X \mid \Gamma , \Delta , \Lambda \vdash_{\F} P \ot D}{
        \infer[\mathsf{sw}]{A \lolli X \mid \Gamma , \Delta \vdash^{\bullet}_{\RI} P}{
          \deduce{A \lolli X \mid \Gamma , \Delta \vdash^{\bullet}_{\F} P}{??}
        }
        &
        \deduce{{-} \mid \Lambda \vdash_{\RI} D}{h}
      }
    \\
    (\text{same derivation as in (\ref{eq:llefttr})})
    &
    (\text{$\lleft$ not applicable since $\Gamma$ is tag-free})
   \end{array}
  \end{displaymath}

 The extra restrictions on sequents and formulae that the rules in (\ref{eq:focus}) impose in comparison to those in (\ref{eq:naive}) do not reduce derivability, e.g., the sequents ${-} \mid Y \vdash_{\F} (X \lolli X) \ot Y$, whose proof requires passivation of a new formula in the derivation of the first premise of a $\tr$ application, $X \lolli Y \mid Z \vdash_{\RI} (X \lolli Y) \ot Z$, whose proof requires application of $\tr$ before $\lleft$,
 and $X \lolli (Y \ot Z) \mid X \vdash_\F Y \ot Z$, which needs $\lleft$ invoked before $\tr$, are all derivable.
 \begin{displaymath}
 	\footnotesize
   \begin{array}{c}
    \proofbox{
     \infer[\tr]{{-} \mid Y \vdash_{\F} (X \lolli X) \ot Y}{
      \infer[\lright]{{-} \vdash \quad \vdash^{\bullet}_{\RI} X \lolli X}{
        \infer[\mathsf{sw}]{{-} \mid X^{\bullet} \vdash^{\bullet}_{\RI} X}{
          \infer[\pass]{{-} \mid X^{\bullet} \vdash^{\bullet}_{\Pass} X}{
            \infer[\mathsf{sw}]{X \mid \quad \vdash_{\LI} X}{
              \infer[\ax]{X \mid \quad \vdash_{\F} X}{}
            }
          }
        }
      }
      &
      \infer[\mathsf{sw}]{{-} \mid Y \vdash_{\RI} Y}{
        \infer[\pass]{{-} \mid Y \vdash_{\Pass} Y}{
          \infer[\mathsf{sw}]{Y \mid \quad \vdash_{\LI} Y}{
            \infer[\ax]{Y \mid \quad \vdash_{\F} Y}{}
          }
        }
      }
     }
     }
 \\
   \proofbox{
     \infer[\tr]{X \lolli Y \mid Z \vdash_{\F} (X \lolli Y) \ot Z}{
      \infer[\lright]{X \lolli Y \mid \quad \vdash^{\bullet}_{\RI} X \lolli Y}{
        \infer[\mathsf{sw}]{X \lolli Y \mid X^{\bullet} \vdash^{\bullet}_{\RI} Y}{
          \infer[\lleft]{X \lolli Y \mid X^{\bullet} \vdash^{\bullet}_{\F} Y}{
            \infer[\mathsf{sw}]{{-} \mid X \vdash_{\RI} X}{
              \infer[\pass]{{-} \mid X \vdash_{\Pass} X}{
                \infer[\mathsf{sw}]{X \mid \quad \vdash_{\LI} X}{
                  \infer[\ax]{X \mid \quad \vdash_{\F} X}{}
                }
              }
            }
            &
            \infer[\mathsf{sw}]{Y \mid \quad \vdash_{\LI} Y}{
              \infer[\ax]{Y \mid \quad \vdash_{\F} Y}{}
            }
          }
        }
      }
      &
      \infer[\mathsf{sw}]{{-} \mid Z \vdash_{\RI} Z}{
        \infer[\pass]{{-} \mid Z \vdash_{\Pass} Z}{
          \infer[\mathsf{sw}]{Z \mid \quad \vdash_{\LI} Z}{
            \infer[\ax]{Z \mid \quad \vdash_{\F} Z}{}
          }
        }
      }
     }
     \quad
     \infer[\lleft]{X \lolli (Y \ot Z) \mid X \vdash_{\F} Y \ot Z}{
      \infer[\mathsf{sw}]{{-} \mid X \vdash_{\RI} X}{
        \infer[\pass]{{-} \mid X \vdash_{\Pass} X}{
          \infer[\mathsf{sw}]{X \mid \quad \vdash_{\LI} X}{
            \infer[\ax]{X \mid \quad \vdash_{\F} X}{}
          }
        }
      }
      &
      \infer[\tl]{Y \ot Z \mid \quad \vdash_{\LI} Y \ot Z}{
        \infer[\mathsf{sw}]{Y \mid Z \vdash_{\LI} Y \ot Z}{
          \infer[\tr]{Y \mid Z \vdash_{\F} Y \ot Z}{
            \infer[\mathsf{sw}]{Y \mid \quad \vdash_{\RI}^{\bullet} Y}{
              \infer[\ax]{Y \mid \quad \vdash_{\F}^\bullet Y}{}
            }
            &
            \infer[\mathsf{sw}]{{-} \mid Z \vdash_{\RI} Z}{
              \infer[\pass]{{-} \mid Z \vdash_{\Pass} Z}{
                \infer[\mathsf{sw}]{Z \mid \quad \vdash_{\LI} Z}{
                  \infer[\ax]{Z \mid \quad \vdash_{\F} Z}{}
                }
              }
            }
          }
        }
      }
     }
     }
    \end{array}
   \normalsize
 \end{displaymath}

We should point out that although the focused calculus is free of the undesired non-determinism that the na\"ive attempt (\ref{eq:naive}) suffered from, it is still non-deterministic and this has to be so. In particular, the focused calculus (\ref{eq:focus}) keeps the following two types of non-determinism of the focused calculus of \cite{uustalu:sequent:2021} (see the analysis in \cite{uustalu:proof:nodate}).
We call these types of non-determinism \emph{essential} because they reflect the fact that there are sequents with multiple derivations in \SkNMILL\ that are not $\circeq$-equivalent.

\begin{enumerate}
  \item[1.] In phase $\Pass$, when the stoup is empty, there is a choice of whether to apply $\pass$ or $\F2\Pass$ and sometimes both options lead to a derivation. For example, the sequent $X \mid \I \ot Y \vdash_{\F} X \ot (\I \ot Y)$ has two distinct derivations in the focused system and the corresponding derivations in \SkNMILL\ are not $\circeq$-equivalent.
  \begin{equation*}
  \footnotesize
    \begin{array}{cc}
    \infer[\tr]{X \mid \I \ot Y \vdash_{\F} X \ot (\I \ot Y)}{
    \infer[\mathsf{sw}]{X \mid \quad \vdash^{\bullet}_{\RI} X}{
      \infer[\ax]{X \mid \quad \vdash_{\F}^\bullet X}{}
    }
      &
      \infer[\mathsf{sw}]{{-} \mid \I \ot Y \vdash_{\RI} \I \ot Y}{
        \infer[\pass]{{-} \mid \I \ot Y \vdash_{\Pass} \I \ot Y}{
          \infer[\tl]{\I \ot Y \mid \quad \vdash_{\LI} \I \ot Y}{
            \infer[\unitl]{\I \mid Y \vdash_{\LI} \I \ot Y}{
              \infer[\mathsf{sw}]{{-} \mid Y \vdash_{\LI} \I \ot Y}{
                \infer[\tr]{{-} \mid Y \vdash_{\F} \I \ot Y}{
                  \infer[\mathsf{sw}]{{-} \mid \quad \vdash^{\bullet}_{\RI} \I}{
                    \infer[\unitr]{{-} \mid \quad \vdash_{\F}^\bullet \I}{}
                  }
                  &
                  \infer[\mathsf{sw}]{{-} \mid Y \vdash_{\RI} Y}{
                    \infer[\pass]{{-} \mid Y \vdash_{\Pass} Y}{
                      \infer[\mathsf{sw}]{Y \mid \quad \vdash_{\LI} Y}{
                        \infer[\ax]{Y \mid \quad \vdash_{F} Y}{}
                      }
                    }
                  }
                }
              }
            }
          }
        }
      }
    }
    &
    \infer[\tr]{X \mid \I \ot Y \vdash_{\F} X \ot (\I \ot Y)}{
      \infer[\mathsf{sw}]{X \mid \quad \vdash^{\bullet}_{\RI} X}{
        \infer[\ax]{X \mid \quad \vdash_{\F}^\bullet X}{}
      }
      &
      \infer[\mathsf{sw}]{{-} \mid \I \ot Y \vdash_{\RI} \I \ot Y}{
        \infer[\tr]{{-} \mid \I \ot Y \vdash_{\F} \I \ot Y}{
          \infer[\mathsf{sw}]{{-} \mid \quad \vdash^{\bullet}_{\RI} \I}{
            \infer[\unitr]{{-} \mid \quad \vdash_{\F}^\bullet \I}{}
          }
          &
          \infer[\mathsf{sw}]{{-} \mid \I \ot Y \vdash_{\RI} Y}{
            \infer[\pass]{{-} \mid \I \ot Y \vdash_{\Pass} Y}{
              \infer[\mathsf{sw}]{\I \ot Y \mid \quad \vdash_{\LI} Y}{
                \infer[\tl]{\I \ot \mid \quad \vdash_{\LI} Y}{
                  \infer[\unitl]{\I \mid Y \vdash_{\LI}}{
                    \infer[\mathsf{sw}]{{-} \mid Y \vdash_{\LI} Y}{
                      \infer[\pass]{{-} \mid Y \vdash_{\Pass} Y}{
                        \infer[\mathsf{sw}]{Y \mid \quad \vdash_{\LI} Y}{
                          \infer[\ax]{Y \mid \quad \vdash_{\F} Y}{}
                        }
                      }
                    }
                  }
                }
              }
            }
          }
        }
      }
    }
    \end{array}
    \normalsize
  \end{equation*}
  \item[2.] In phase $\F$, if the succedent formula is $A \ot B$, and the rule $\tr$ is to be applied, the context can be split anywhere. Sometimes
several of these splits can lead to a derivation.
  For example, the sequent $X \mid \I , Y \vdash_{\F} (X \ot \I) \ot Y$ has two distinct derivations.
  \begin{equation*}
  \footnotesize
    \begin{array}{cc}
      \infer[\tr]{X \mid \I , Y \vdash_{\F} (X \ot \I) \ot Y}{
       \infer[\mathsf{sw}]{X \mid \I \vdash^{\bullet}_{\RI} X \ot \I}{
         \infer[\tr]{X \mid \I \vdash_{\F}^\bullet X \ot \I}{
           \infer[\mathsf{sw}]{X \mid \quad \vdash^{\bullet}_{\RI} X}{
             \infer[\ax]{X \mid \quad \vdash_{\F}^\bullet X}{}
           }
           &
           \infer[\mathsf{sw}]{{-} \mid \I \vdash_{\RI} \I}{
             \infer[\pass]{{-} \mid \I \vdash_{\Pass} \I}{
               \infer[\unitl]{\I \mid \quad \vdash_{\LI} \I}{
                 \infer[\mathsf{sw}]{{-} \mid \quad \vdash_{\LI} \I}{
                   \infer[\unitr]{{-} \mid \quad \vdash_{\F} \I}{}
                 }
               }
             }
           }
         }
       }
       &
       \infer[\mathsf{sw}]{{-} \mid Y \vdash_{\RI} Y}{
         \infer[\pass]{{-} \mid Y \vdash_{\Pass} Y}{
           \infer[\mathsf{sw}]{Y \mid \quad \vdash_{\LI} Y}{
             \infer[\ax]{Y \mid \quad \vdash_{\F} Y}{}
           }
         }
       }
      }
      &
      \infer[\tr]{X \mid \I , Y \vdash_{\F} (X \ot \I) \ot Y}{
       \infer[\mathsf{sw}]{X \mid \quad \vdash^{\bullet}_{\RI} X \ot \I}{
         \infer[\tr]{X \mid \quad \vdash_{\F}^\bullet X \ot \I}{
           \infer[\mathsf{sw}]{X \mid \quad \vdash^{\bullet}_{\RI} X}{
             \infer[\ax]{X \mid \quad \vdash_{\F}^\bullet X}{}
           }
           &
           \infer[\mathsf{sw}]{{-} \mid \quad \vdash_{\RI} \I}{
             \infer[\ax]{{-} \mid \quad \vdash_{\F} \I}{}
           }
         }
       }
       &
       \infer[\mathsf{sw}]{{-} \mid \I , Y \vdash_{\RI} Y}{
         \infer[\pass]{{-} \mid \I , Y \vdash_{\Pass} Y}{
           \infer[\unitl]{\I \mid Y \vdash_{\LI} Y}{
             \infer[\mathsf{sw}]{{-} \mid Y \vdash_{\LI} Y}{
               \infer[\pass]{Y \mid \quad \vdash_{\Pass} Y}{
                 \infer[\mathsf{sw}]{Y \mid \quad \vdash_{\LI} Y}{
                   \infer[\ax]{Y \mid \quad \vdash_{\F} Y}{}
                 }
               }
             }
           }
         }
       }
      }
    \end{array}
    \normalsize
  \end{equation*}
\end{enumerate}

The presence of $\lolli$ adds two further types of essential
nondeterminism.

\begin{enumerate}
\item[3.] This type is similar to type 2. In phase $\F$, if the stoup formula is $A \lolli B$, and the rule $\lleft$ is to be applied, the context can be split anywhere. Again, sometimes several of these splits lead to a derivation.
 For example, the sequent $\I \lolli (X \lolli Y) \mid \I , X \vdash_{\F} Y$ has two derivations.
\begin{equation*}
\footnotesize
  \begin{array}{cc}
    \infer[\lleft]{\I \lolli (X \lolli Y) \mid \I , X \vdash_{\F} Y}{
      \infer[\mathsf{sw}]{{-} \mid \I \vdash_{\RI} \I}{
        \infer[\pass]{{-} \mid \I \vdash_{\Pass} \I}{
          \infer[\unitl]{\I \mid \quad \vdash_{\LI} \I}{
            \infer[\mathsf{sw}]{{-} \mid \quad \vdash_{\LI} \I}{
              \infer[\unitr]{{-} \mid \quad \vdash_{\F} \I}{}
            }
          }
        }
      }
      &
      \infer[\mathsf{sw}]{X \lolli Y \mid X \vdash_{\LI} Y}{
        \infer[\lleft]{X \lolli Y \mid X \vdash_{\F} Y}{
          \infer[\mathsf{sw}]{{-} \mid X \vdash_{\RI} X}{
            \infer[\pass]{{-} \mid X \vdash_{\Pass} X}{
              \infer[\mathsf{sw}]{X \mid \quad \vdash_{\LI} X}{
                \infer[\ax]{X \mid \quad \vdash_{\F} X}{}
              }
            }
          }
          &
          \infer[\mathsf{sw}]{Y \mid \quad \vdash_{\LI} Y}{
            \infer[\ax]{Y \mid \quad \vdash_{\F} Y}{}
          }
        }
      }
    }
    &
    \infer[\lleft]{\I \lolli (X \lolli Y) \mid \I , X \vdash_{\F} Y}{
      \infer[\mathsf{sw}]{{-} \mid \quad \vdash_{\RI} \I}{
        \infer[\unitr]{{-} \mid \quad \vdash_{\F} \I}{}
      }
      &
      \infer[\mathsf{sw}]{X \lolli Y \mid \I , X \vdash_{\LI} Y}{
        \infer[\lleft]{X \lolli Y \mid \I , X \vdash_{\F} Y}{
          \infer[\mathsf{sw}]{{-} \mid \I , X \vdash_{\RI} X}{
            \infer[\pass]{{-} \mid \I , X \vdash_{\Pass} X}{
              \infer[\unitl]{\I \mid X \vdash_{\LI} X}{
                \infer[\mathsf{sw}]{{-} \mid X \vdash_{\LI} X}{
                  \infer[\pass]{{-} \mid X \vdash_{\Pass} X}{
                    \infer[\mathsf{sw}]{X \mid \quad \vdash_{\LI} X}{
                      \infer[\ax]{X \mid \quad \vdash_{\F} X}{}
                    }
                  }
                }
              }
            }
          }
          &
          \infer[\mathsf{sw}]{Y \mid \quad \vdash_{\LI} Y}{
            \infer[\ax]{Y \mid \quad \vdash_{\F} Y}{}
          }
            }
          }
        }
  \end{array}
  \normalsize
\end{equation*}
\item[4.] Finally, in phase $\F$, if the succedent formula is $A \ot B$ and the stoup formula is $A' \lolli B'$, then both $\tr$ and $\lleft$ can be applied first and sometimes both options lead to a derivation. For example, the sequent 
$\I \lolli \I \mid Z \vdash_{\F} (\I \lolli \I) \ot Z$ has two derivations. 
\begin{equation*}
\footnotesize
  \begin{array}{cc}
    \infer[\tr]{\I \lolli \I \mid Z \vdash_{\F} (\I \lolli \I) \ot Z}{
      \infer[\lright]{\I \lolli \I \mid \quad \vdash^{\bullet}_{\RI} \I \lolli \I}{
        \infer[\mathsf{sw}]{\I \lolli \I \mid \I^{\bullet} \vdash^{\bullet}_{\RI} \I}{
          \infer[\lleft]{\I \lolli \I \mid \I^{\bullet} \vdash^{\bullet}_{\F} \I}{
            \infer[\mathsf{sw}]{{-} \mid \I \vdash_{\RI} \I}{
              \infer[\pass]{{-} \mid \I \vdash_{\Pass} \I}{
                \infer[\unitl]{\I \mid \quad \vdash_{\LI} \I}{
                  \infer[\mathsf{sw}]{{-} \mid \quad \vdash_{\LI} \I}{
                    \infer[\unitr]{{-} \mid \quad \vdash_{\F} \I}{}
                  }
                }
              }
            }
            &
            \infer[\unitl]{\I \mid \quad \vdash_{\LI} \I}{
              \infer[\mathsf{sw}]{{-} \mid \quad \vdash_{\LI} \I}{
                \infer[\unitr]{{-} \mid \quad \vdash_{\F} \I}{}
              }
            }
          }
        }
      }
      &
      \infer[\mathsf{sw}]{{-} \mid Z \vdash_{\LI} Z}{
        \infer[\pass]{{-} \mid Z \vdash_{\Pass} Z}{
          \infer[\mathsf{sw}]{Z \mid \quad \vdash_{\LI} Z}{
            \infer[\ax]{Z \mid \quad \vdash_{\F} Z}{}
          }
        }
      }
    }
    &
    \infer[\lleft]{\I \lolli \I \mid Z \vdash_{\F} (\I \lolli \I) \ot Z}{
      \infer[\mathsf{sw}]{{-} \mid \quad \vdash_{\RI} \I}{
        \infer[\unitr]{{-} \mid \quad \vdash_{\F} \I}{}
      }
      &
      \infer[\unitl]{\I \mid Z \vdash_{\LI} (\I \lolli \I) \ot Z}{
        \infer[\mathsf{sw}]{{-} \mid Z \vdash_{\LI} (\I \lolli \I) \ot Z}{
          \infer[\tr]{{-} \mid Z \vdash_{\F} (\I \lolli \I) \ot Z}{
            \infer[\lright]{{-} \mid \quad \vdash^{\bullet}_{\RI} \I \lolli \I}{
              \infer[\mathsf{sw}]{{-} \mid \I^{\bullet} \vdash^{\bullet}_{\RI} \I}{
                \infer[\pass]{{-} \mid \I^{\bullet} \vdash^{\bullet}_{\Pass} \I}{
                  \infer[\mathsf{sw}]{\I \mid \quad \vdash_{\LI} \I}{
                    \infer[\unitl]{\I \mid \quad \vdash_{\LI} \I}{
                      \infer[\mathsf{sw}]{{-} \mid \quad \vdash_{\LI}}{
                        \infer[\unitr]{{-} \mid \quad \vdash_{\F} \I}{}
                      }
                    }
                  }
                }
              }
            }
            &
            \infer[\mathsf{sw}]{{-} \mid Z \vdash_{\LI} Z}{
              \infer[\pass]{{-} \mid Z \vdash_{\Pass} Z}{
                \infer[\mathsf{sw}]{Z \mid \quad \vdash_{\LI} Z}{
                  \infer[\ax]{Z \mid \quad \vdash_{\F} Z}{}
                }
              }
            }
          }
        }
      }
    }
  \end{array}
  \normalsize
\end{equation*}
Note that, in the second derivation, the rule $\pass$ applies to the sequent ${-} \mid \I^{\bullet} \vdash^{\bullet}_{\Pass} \I$ only because the context formula $\I^\bullet$ is tagged.
\end{enumerate}


 \begin{theorem}
   The focused sequent calculus is sound and complete wrt.\ the  sequent calculus of Section \ref{sec2}: there is a bijective correspondence between the set of derivations of $S \mid \Gamma \vdash A$ quotiented by $\circeq$ and the set of derivations of $S \mid \Gamma \vdash_\RI A$.
 \end{theorem}
Soundness is immediate: there exist functions $\mathsf{emb}_{ph} : S \mid \Gamma \vdash^x_{ph} A \to S \mid \Gamma \vdash A$, for all $ph \in \{\RI,\LI,\Pass, \F \}$, which simply erase all phase annotations and tags. Completeness follows from the fact that the following rules are all admissible:
\begin{equation}\label{eq:admis}
\small
  \begin{array}{c}
    \infer[\unitl^{\RI}]{\I \mid \Gamma \vdash_{\RI} C}{{-} \mid \Gamma \vdash_{\RI} C}
    \quad
    \infer[\tl^{\RI}]{A \ot B \mid \Gamma \vdash_{\RI} C}{A \mid B, \Gamma \vdash_{\RI} C}
    \quad
    \infer[\pass^{\RI}]{{-} \mid \Gamma \vdash_{\RI} C}{A \mid \Gamma \vdash_{\RI} C}
    \quad
    \infer[\ax^{\RI}]{A \mid \quad \vdash_{\RI} A}{}
    \quad
    \infer[\unitr^{\RI}]{{-} \mid \quad \vdash_{\RI} \I}{}
\\[6pt]
    \infer[\lleft^{\RI}]{A \lolli B \mid \Gamma , \Delta \vdash_{\RI} C}{
    {-} \mid \Gamma \vdash_{\RI} A
    &
    B \mid \Delta \vdash_{\RI} C
    }
    \qquad
    \infer[\tr_{\Gamma'}^{\RI}]{S \mid \Gamma , \Delta \vdash_{\RI} \ldbc \Gamma' \mid A \rdbc_{\lolli} \ot B}{
      S \mid \Gamma , \Gamma' \vdash_{\RI} A
      &
      {-} \mid \Delta \vdash_{\RI} B
    }
  \end{array}
\end{equation}
The interesting one is $\tr_{\Gamma'}^{\RI}$. The tensor right rule $\tr^\RI$, with premises and conclusion in phase $\RI$, is an instance of the latter with empty $\Gamma'$.
Without this generalization including the extra context $\Gamma'$, one quickly discovers that finding a proof of $\tr^\RI$, proceeding by induction on the structure of the derivation of the first premise, is not possible when this derivation ends with an application of $\lright$:
\begin{displaymath}
\small
  \proofbox{
    \infer[\tr^{\RI}]{S \mid \Gamma , \Delta \vdash_{\RI} (A' \lolli B') \ot B}{
    \infer[\lright]{S \mid \Gamma \vdash_{\RI} A' \lolli B'}{
      \deduce{S \mid \Gamma , A' \vdash_{\RI} B'}{f}
    }
    &
    \deduce{{-} \mid \Delta \vdash_{\RI} B}{g}
    }
    } = \quad ??
\end{displaymath}
The inductive hypothesis applied to $f$ and $g$ would produce a derivation of the wrong sequent. The use of $\Gamma'$ in the generalized rule $\tr_{\Gamma'}^{\RI}$ is there to fix precisely this issue.

The admissibility of the rules in (\ref{eq:admis}) allows the construction of a function $\mathsf{focus} : S \mid \Gamma \vdash A \to \linebreak S \mid \Gamma \vdash_\RI A$, replacing applications of each rule in (\ref{eq:seqcalc}) with inferences by the corresponding admissible focused rule in phase $\RI$.
It is possible to prove that the function $\mathsf{focus}$ maps  $\circeq$-equivalent derivations in the sequent calculus for \SkNMILL\ to syntactically identical derivations in focused sequent calculus. One can also show that $\mathsf{focus}$ is the inverse of $\mathsf{emb}_\RI$, i.e. $\mathsf{focus}\;(\mathsf{emb}_\RI \;f) = f$ for all $f : S \mid \Gamma \vdash_\RI A$ and $\mathsf{emb}_\RI\;(\mathsf{focus}\;g) \circeq g$ for all $g : S \mid \Gamma \vdash A$. In other words, each $\circeq$-equivalence class in the sequent calculus corresponds uniquely to a single derivation in the focused sequent calculus.

The focused sequent calculus solves the \emph{coherence problem} for skew monoidal closed categories.
As proved in Theorems \ref{thm:models}, \ref{thm:unique}, the sequent calculus for \SkNMILL\ is a presentation of the free skew monoidal closed category $\FSkMCC(\mathsf{At})$ on the set $\mathsf{At}$. The coherence problem is the problem of deciding whether two parallel maps in $\FSkMCC(\mathsf{At})$ are equal. 
This is equivalent to deciding whether two sequent calculus derivations $f,g : A \mid ~ \vdash B$ are in the same $\circeq$-equivalence class. But that in turn is the same as deciding whether $\mathsf{focus}\;f = \mathsf{focus}\;g$ in the focused sequent calculus, and deciding syntactic equality of focused derivations is straightforward. The Hilbert-style calculus is a direct presentation of $\FSkMCC(\mathsf{At})$, but thanks to the bijection (up to $\doteq$ resp.\ $\circeq$) between Hilbert-style and sequent calculus derivations, we can also decide if two Hilbert-style derivations $f, g : A \Rightarrow B$ are in the same $\doteq$-equivalence class.

\section{Conclusion}

The paper describes a sequent calculus for \SkNMILL, a skew variant of non-commutative multiplicative intuitionistic linear logic. The introduction of the logic \SkNMILL\ via this sequent calculus is motivated by the categorical notion of skew monoidal closed category, which yields the intended categorical semantics for the logic. Sequent calculus derivations admit unique normal forms wrt.\ a congruence relation $\circeq$ capturing the equational theory of skew monoidal closed categories at the level of derivations. Normal forms can be organized in a focused sequent calculus, where each focused derivation uniquely corresponds to (and so represents) a $\circeq$-equivalence class in the unfocused sequent calculus. In order to deal with all the permutative conversions of $\circeq$ and to consequently eliminate the related sources of non-determinism in root-first proof search, the focused sequent calculus employs a system of tags for keeping track of the new formulae occurring in the context while building a derivation. The focused sequent calculus solves the coherence problem for skew monoidal closed categories: deciding if a canonical diagram commutes in every skew monoidal closed category reduces to deciding equality of focused derivations.

We plan to investigate alternative presentations of \SkNMILL, such as natural deduction. Analogously to the case of the sequent calculus studied in this paper, we expect natural deduction derivation to be strongly normalizing wrt.\ an appropriately defined $\beta\!\eta$-conversion. We are interested in directly comparing the resulting $\beta\!\eta$-long normal forms with the focused derivations of Section \ref{sec:focus}.

The system of tags in focused proofs, used for taming the non-deterministic choices arising in proof search to be able to canonically represent each equivalence class of $\circeq$, appears to be a new idea. 
There exist other techniques for drastically reducing non-determinism in proof search, such as multi-focusing \cite{chaudhuri:canonical:2008} and saturated focusing \cite{scherer:simple:2015}, and we wonder if our system of tags is in any way related to these. As a general disclaimer, we should state that we have interpreted focusing broadly as the idea of root-first proof search defining a normal form, based on a careful discipline of application of invertible and non-invertible rules to reduce non-determinism, but not necessarily driven by a polarity-centric analysis. In this respect, our approach is similar in spirit to, e.g., \cite{dyckhoff:ljq}.

This paper represents the latest installment of a large project aiming at the development of the proof theory of categories with skew structure. So far the project, promoted and advanced by Uustalu, Veltri and Zeilberger, has investigated proof systems for skew semigroup categories \cite{zeilberger:semiassociative:19}, (non-symmetric and symmetric) skew monoidal categories \cite{uustalu:sequent:2021,uustalu:proof:nodate,veltri:coherence:2021} and skew prounital closed categories \cite{uustalu:deductive:nodate}. From a purely proof-theoretic perspective, the main lesson gained from the study of these skew systems consists in first insights into ways of \emph{modular} construction of focusing calculi. This is particularly highlighted in Uustalu et al.'s study of \emph{partially normal} skew monoidal categories \cite{uustalu:proof:nodate}, where one or more structural laws among $\lambda$, $\rho$ and $\alpha$ can be required to be invertible. The focused sequent calculus of partially normal skew monoidal categories is obtained from the focused sequent calculus of skew monoidal categories by modularly adding new rules for each enforced normality condition.
We expect similar modularity to show up in the case of partially normal skew monoidal closed categories.

We are also planning to develop an extension of \SkNMILL\ with an exponential modality $!$. This in turns requires the study of linear exponential comonads on skew monoidal closed categories, extending the work of Hasegawa \cite{hasegawa:linear:2017}. We would also like to consider modalities for exchange \cite{jiang:lambek:2019} and associativity/unitality.

\paragraph{Acknowledgements} N.V.\ and T.U.\ were supported by the
Estonian Research Council grants no. \linebreak PSG659, PSG749 and PRG1210, N.V.\ and
C.-S.W.\ by the ESF funded Estonian IT Academy research measure
(project 2014-2020.4.05.19-0001). T.U.\ was supported by the Icelandic
Research Fund grant no.~196323-053. Participation of C.-S.W.\ in the conference was supported by the EU
COST action CA19135 (CERCIRAS).

  \bibliographystyle{eptcs}
  \bibliography{biblio}
\end{document}